\documentclass{article}
\usepackage{amsmath}
\usepackage{amsthm}
\usepackage{todonotes}
\usepackage{ amssymb }

\theoremstyle{plain} 
\newtheorem{theorem}{Theorem}
\newtheorem{lemma}[theorem]{Lemma}

\newtheorem{corollary}[theorem]{Corollary}

\theoremstyle{definition}

\theoremstyle{remark}
\newtheorem{remark}[theorem]{Remark}
\newtheorem{example}[theorem]{Example}
\newcommand{\ttrue}{\mathit{true}}
\newcommand{\ffalse}{\mathit{false}}

\newcommand{\truef}{\textnormal{true}}

\newcommand{\adom}[1]{\ensuremath\textnormal{adom}(#1)}

\newcommand{\setqueries}[1]{\ensuremath #1}

\newcommand{\lc}[1]{\ensuremath \setqueries{#1}^{\subseteq}}
\newcommand{\lns}[1]{\ensuremath \setqueries{#1}^{\neq\emptyset}}

\newcommand{\MONO}{\ensuremath\textnormal{MON}}

\newcommand{\CON}{\ensuremath \textnormal{CQ}}

\newcommand{\SCON}{\ensuremath\CON_{\dbsc{}}}

\newcommand{\dbsc}{\ensuremath\Gamma}
\newcommand{\restr}[2]{#1|_{#2}}
 \title{A Monotone Preservation Result for Boolean Queries Expressed as a Containment of Conjunctive Queries}
\DeclareMathOperator{\dom}{adom}
\newcommand{\seq}[3]{\ensuremath ({#1}_{#2})_{#2 \in #3}}
\date{Hasselt University}
\author{Dimitri Surinx \and Jan Van den Bussche}
\begin{document}
\maketitle
\begin{abstract}
  When a relational database is queried, the result is normally a relation.
  Some queries, however, only require a yes/no answer; such
  queries are often called \emph{boolean queries}. It is customary in database theory
  to express boolean queries by testing nonemptiness of query expressions. Another interesting way for expressing boolean queries are containment statements of 
  the form $Q_1 \subseteq Q_2$ where $Q_1$ and $Q_2$ are query expressions. Here, for any input
  instance $I$, the boolean query result is $\ttrue$ if $Q_1(I)$ is a subset of $Q_2(I)$ and $\ffalse$ otherwise.

  In the present paper we will focus on nonemptiness and containment statements about conjunctive queries. The main goal is to investigate the monotone fragment of the containments of conjunctive queries. In particular, we show a preservation like result for this monotone fragment. That is, we show that, in expressive power, the monotone containments of conjunctive queries are exactly equal to conjunctive queries under nonemptiness.
\end{abstract}
\section{Introduction}
In this paper, we compare boolean queries (or integrity constraints) expressed using conjunctive
queries (CQs~\cite{ahv_book}) 
in two different ways: 
\begin{description}
  \item[Nonemptiness:] As an expression of the form $Q \neq \emptyset$, with $Q$ a CQ;
  \item[Containment:] As an expression of the form $Q_1\subseteq Q_2$, with $Q_1$ and $Q_2$ two CQs.
\end{description}

An example of a nonemptiness query is ``there exists a customer who bought a luxury product''. An
example of a containment query is ``every customer who bought a luxury product also bought a sports
product''. A qualitative difference between nonemptiness and containment queries is that
nonemptiness queries are always monotone: when the result is true on some input instance, it is also
true on any larger instance. In contrasts, containment queries need not be monotone, as shown by
the example above. The nonemptiness of a CQ is always expressible as the containment of two CQs. For example, the nonemptiness of $(x) \leftarrow \mathit{Customer}(x),
\mathit{Bought(x,y)}, \mathit{Luxury}(y)$ is expressed as
\[ ()\leftarrow \mathit{true}\quad \subseteq \quad () \leftarrow \mathit{Customer}(x),
\mathit{Bought(x,y)}, \mathit{Luxury}(y).\]
Conversely, one may suspect that, as far as \emph{monotone} queries are concerned, nothing more is
expressible by a containment of two CQs. Indeed, we show in this paper that every monotone query
expressed as the containment of two CQs is already expressible as the nonemptiness of a CQ.
Such a result fits the profile of a preservation theorem since it gives a syntactical language for a
semantical sublanguage. Preservation theorems have been studied intensively in model theory, finite model theory and database theory~\cite{gurevich_complexity,gurevich_monpos,changkeisler,stolboushkin_finmon,rossman_hom,benedikt_book}.

From our proof it also follows that monotonicity testing of a containment of two CQs is decidable;
specifically, the problem is 
NP-complete.
\section{Preliminaries}
A database schema $\dbsc$ is a finite nonempty set of relation names.
Every relation name $R$ is assigned an arity, which is a natural
number. Let $V$ be some fixed infinite universe of data elements and let $R$ be a relation name of
arity $n$. An \emph{$R$-fact} is an expression of the form $R(a_1,\ldots,a_n)$ where $a_i\in V$ for
$i=1,\ldots,n$. Generally, a \emph{fact} is an $R$-fact for
some $R$. An \emph{$R$-instance} $I$ is a finite set of $R$-facts. 
More generally, an instance $I$ of
a database schema $\dbsc$ is defined to be a nonempty union $\bigcup_{R \in \dbsc{}} I(R)$, where $I(R)$ is
an $R$-instance. The active domain of an instance $I$, denoted by $\adom{I}$, is the
set of all data elements from $V$ that occur in $I$. An instance $I$ is called $\emph{connected}$
when for every two data elements $a,b\in \adom{I}$ there is a sequence of facts $f_1,\ldots,f_n$ in $I$ such
that: $a$ is in $\adom{\{f_1\}}$, $b$ is in $\adom{\{f_n\}}$, and $\adom{\{f_i\}}\cap \adom{\{f_
{i+1}\}}\neq \emptyset$ for any $i=1,\ldots,n-1$. 
  An instance $J$ is a called a \emph{connected component} of $I$ if $J$
is connected, $J\subseteq I$ and $J$ is maximal in $I$ with respect to inclusion.

We have defined database and instances under the so called ``logic programming perspective''~
\cite{ahv_book}. We will define the results of conjunctive queries, however, under the so-called
``named'' perspective~\cite{ahv_book}. This will allow a lighter notation in our proof of Lemma~\ref{lem:contCQMONemptyheads} where we are taking subtuples of heads of conjunctive queries. 

 In the named perspective, \emph{tuples} are defined over a finite set
of attributes, which we refer to as a \emph{relation scheme}. Formally, \emph{tuples}, say $t = 
\seq{u}{i}{S}$ on a relation scheme $S$, are considered as mappings, so $t$ is a mapping on $S$ and
$t(i)=u_i$.  Then, subtuples, say $\restr{t}{K}$ for $K\subseteq S$ are treated as restrictions of
the mapping $H$ to $K$. On the empty relation scheme, there is only one tuple, namely the empty
mapping, also called the empty tuple. We denote the empty tuple by $()$.

We formalize the notion of conjunctive queries as follows. From the outset we assume an infinite universe of
variables. A \emph{conjunctive query} is an expression of the form $Q: H\leftarrow B$ where the head
$H$ is a tuple of variables (tuple in the sense as just defined), and the body $B$ is a set of atoms
over $\dbsc$. An \emph{atom} is an expression of the form $R(v_1,\ldots,v_n)$ where $R \in \dbsc$ and $v_1,\ldots,v_n$ are variables. We will denote the set of conjunctive queries over $\dbsc$ as $\SCON$. For a conjunctive query $Q$ we will write $H_{Q}$ for the head and $B_{Q}$ for the body of $Q$. The \emph{result scheme} of a conjunctive query $Q$ is the relation scheme of the head $H_Q$. Note that we allow \emph{unsafe} queries, i.e., queries with head variables that do not appear in the body. Semantically, for any instance $I$ over $\dbsc$, $Q(I)$ is defined as:
\[\{f\circ H_Q \mid \text{$f$ is a homomorphism from $Q$ into $I$}\}. \]
 
Here, a homomorphism $f$ from $Q$ into $I$ is a function on the variables in $H_Q$ and $B_{Q}$ to $
\adom{I}$ such that $f(B_Q)\subseteq I$. When the variables in $H_Q$ are all present in $B_{Q}$, we
will also write that $f$ is a homomorphism from $B_{Q}$ into $I$. Interchangeably, we will write that $B_{Q}$ maps into $I$. 

\begin{example}
  Consider the database schema with the relation name $\mathit{Flights}$ of arity two. The
  following conjunctive query returns all the city pairs that are connected by flight with 
  one stopover:
  \[(A : x, B : y) \leftarrow \mathit{Flight}(x,z),\mathit{Flights}(z,y).\]

  This query returns $\{(A : \mathit{Vienna},B :\mathit{Brussels}), (A : \mathit{Paris},B :
  \mathit{Rome})\}$ on the instance \[\{\mathit{Flights(Paris,Brussels)},\mathit{Flights(Brussels,Rome)}
  , \mathit{Flights(Vienna,Paris)}\}.\]
\end{example}

\begin{remark}
  It is convenient to assume that variables are data elements in $V$. Then, we can use the body of a conjunctive query as a database instance. As a consequence, an $R$-atom can then be thought of as an $R$-fact.
\end{remark}

For any two queries $Q_1$ and $Q_2$, we write $Q_1\sqsubseteq Q_2$ if $Q_1(I)\subseteq Q_2(I)$ for
any database instance $I$ over $\dbsc$. We recall:
\begin{theorem}[\cite{cm}]\label{theo:cm}
  Let $Q_1$ and $Q_2$ be conjunctive queries. Then, $Q_1\sqsubseteq Q_2$ iff $H_{Q_1}\in Q_2(B_{Q_1})$.
\end{theorem}

A \emph{boolean query} over a database schema $\dbsc$ is a mapping from
instances of $\dbsc$ to $\{\ttrue,\ffalse\}$. We can associate to any conjunctive query $Q$, a
boolean query $Q \neq \emptyset$, that is $\ttrue$ on $I$ if $Q(I)\neq \emptyset$ and $\ffalse$ if
$Q(I)=\emptyset$. We will write $\lns{\SCON}$ for the family of boolean queries of the form $Q\neq
\emptyset$ where $Q$ is in $\SCON$.

As argued in the introduction, this is not the only natural way to express boolean queries.
Containment statements of the form $Q_1\subseteq Q_2$ provide a clean way to express interesting
nonmonotone boolean queries. Formally, the boolean query $Q_1\subseteq Q_2$ is $\ttrue$ on $I$ if
$Q_1(I)$ is a subset of $Q_2(I)$, and $\ffalse$ on $I$ otherwise. It is understood that we can only take containment boolean queries of two conjunctive queries $Q_1$ and $Q_2$ if they have the same result scheme. We write $\lc{\SCON}$ for the family of boolean queries expressible by containment statements $Q_1\subseteq Q_2$ where $Q_1$ and $Q_2$ are in $\SCON$ with the same result scheme.

Recall that every conjunctive query $Q$ is $\emph{monotone}$, in the sense that for any two instances $I,J$ over $\dbsc$, such that $I\subseteq J$, we have $Q(I)\subseteq Q(J)$. Furthermore, we say that a boolean query $Q$ is \emph{monotone} if for any two instances $I,J$ over $\dbsc$, such that $I\subseteq J$, we have $Q(I)=\ttrue$ implies $Q(J)=\ttrue$. We denote the set of monotone boolean queries with $\MONO$.

We will frequently use the following property of conjunctive queries with connected bodies. If $Q$ is a conjunctive query with a connected body, then $Q(I\cup J)= Q(I)\cup Q(J)$ for any domain-disjoint instances $I$ and $J$. We will refer to this property as the \emph{additivity} property. Furthermore, we say that a query $Q$ is additive if it has the additivity property.
\section{Main result} 
In this section we will prove the main theorem of the present paper. This preservation theorem can be summarized as follows:
\begin{theorem}\label{theo:presv}
 For any database schema $\dbsc{}$, $\lc{\SCON}\cap \MONO= \lns{\SCON}$. Specifically, every
 monotone query $Q_1\subseteq Q_2$, where $Q_1$ and $Q_2$ are CQs, is equivalent to a query of the form $(()\leftarrow B)\neq
 \emptyset$, where $B$ is empty or $B$ consists of some of the connected components of $B_{Q_2}$. 
\end{theorem}

Note that $\lns{\SCON}\subseteq \lc{\SCON}\cap \MONO$ already follows from the fact that $Q \neq
\emptyset$ is equivalent to $()\leftarrow \emptyset \subseteq ()\leftarrow B_Q$~. To prove the remaining inclusion we first establish a few
technical results. First, we show that any monotone containment of conjunctive queries is equivalent
to a containment of conjunctive queries with empty heads. For the remainder of this section, we write $Z_a$ to be the instance where there is exactly one fact $R(a,a,\ldots,a)$ for every $R \in \dbsc$. Note that for every $\CON$ $Q$, we have $Q(Z_a) = \{(a,a,\ldots,a)\}$.

\begin{lemma}\label{lem:contCQMONemptyheads}
   Let $Q_1$ and $Q_2$ be conjunctive queries. If $Q_1\subseteq Q_2$ is monotone, then it is equivalent to the conjunctive query $()\leftarrow B_{Q_1}\subseteq ()\leftarrow B_{Q_2}$.
  \end{lemma}
  \begin{proof}
   Let $S$ be the result scheme of $Q_1$ and $Q_2$. Write $B_{Q_2}$ as $B_1,\ldots,B_k,B$ where the $B_j$ are the connected components of $B_{Q_2}$ that contain at least one variable in $H_{Q_2}$, and $B$ is the collection of the remaining connected components. 

     Define $A_j = \{i \in S \mid H_{Q_2}(i) \in \dom(B_{j})\}$ for $j=1,\ldots,k$ and let $A_{0}$ contain the remaining attributes in $S$. Furthermore, define $A = \bigcup_{1\leq j \leq k} A_j$.

     We first show that there is a function $f$ such that $f\circ \restr{H_{Q_2}}{A_0}=\restr{H_{Q_1}}{A_0}$. Let $a$ be a fresh data element.
     Define $I = Z_a \cup B_{Q_1} \cup \bigcup_{i\in C} Z_{H_{Q_1}(i)}$ where $C = \{i \in S \mid H_
     {Q_1}(i)\not \in \dom{(B_{Q_1})}\}$. Since, $Q_1(Z_a) = Q_2(Z_a)$ and $Q_1\subseteq Q_2$ is monotone, we have $Q_1(I)\subseteq Q_2(I)$. Therefore, $H_{Q_1}\in Q_2(I)$ since $H_{Q_1}\in Q_1(I)$. Hence, there is a homomorphism $f$ from $Q_2$ into $I$ such that $f\circ H_{Q_2} = H_{Q_1}$. In particular, $f \circ \restr{H_{Q_2}}{A_0}=\restr{H_{Q_1}}{A_0}$ as desired.

    Next, we show for each $j=1,\ldots,k$ that  \[(\restr{H_{Q_1}}{A_j}\leftarrow B_{Q_1})\sqsubseteq (\restr{H_{Q_2}}{A_j}\leftarrow B_{j}).\tag{$\star$}\] 

    Let $I$ be a nonempty instance over $\dbsc$ and let $a$ be a fresh data element. Suppose $t \in (\restr{H_{Q_1}}{A_j}\leftarrow B_{Q_1})(I)$. Since $(\restr{H_{Q_1}}{A_j}\leftarrow B_{Q_1})$ and $Q_1$ have the same body, and $\restr{H_{Q_1}}{A_j}$ is a subtuple of $H_{Q_1}$, we can extend $t$ to $t'$ such that $t' \in Q_1(I)$. Furthermore, since $Q_1\subseteq Q_2$ is monotone and $Q_1(Z_a)= Q_2(Z_a)$, we have $Q_1(I\cup Z_a)\subseteq Q_2(I\cup Z_a)$. Thus, $t'\in Q_2(I\cup Z_a)$, whence we also have $t\in (\restr{H_{Q_2}}{A_j}\leftarrow B_{j})(I \cup Z_a)$. Since $\restr{H_{Q_2}}{A_j}\leftarrow B_{j}$ is additive, $t \in (\restr{H_{Q_2}}{A_j}\leftarrow B_{j})(I)\cup (\restr{H_{Q_2}}{A_j}\leftarrow B_{j})(Z_a)$. This implies that $t \in (\restr{H_{Q_2}}{A_j}\leftarrow B_{j})(I)$ since $t$ is a tuple of data elements in $I$.

 We now show that $Q_1\subseteq Q_2$ is equivalent to $Q_1'\subseteq Q_2'$ where $Q_1' = ()\leftarrow B_{Q_1}$ and $Q_2'=() \leftarrow B_{Q_2}$, which proves our lemma. Clearly, $Q_1(I)\subseteq Q_2(I)$ implies that $Q_1'(I)\subseteq Q_2'(I)$. For the other direction, suppose that $Q_1'(I)\subseteq Q_2'(I)$ and let $t \in Q_1(I)$. Then, we have the following:
\begin{itemize}
  \item There is a homomorphism $f_1$ from $B_{Q_1}$ to $I$ such that $f_1\circ H_{Q_1}=t$.
  \item There is a homomorphism $f_2$ from $B_{Q_2}$ to $I$ since $\emptyset\neq Q_1'(I)\subseteq Q_2'(I)$.
  \item There is a function $h$ such that $h\circ\restr{H_{Q_2}}{A_0} = \restr{H_{Q_1}}{A_0}$. 
  \item For every $j=1,\ldots,k,$ $\restr{t}{A_j}\in (\restr{H_{Q_2}}{A_j}\leftarrow B_{j})(I)$ by $(\star)$.
  Hence, there is a homomorphism $h_j$ from $B_{j}$ into $I$ such that $h_j\circ \restr{H_{Q_2}}{A_j}=\restr{t}{A_j}$.
\end{itemize}

 We now construct a homomorphism $f$ from $Q_2$ into $I$ such that $f\circ H_{Q_2} = t$. 
We define this $f$ as follows:
\[
f: x \mapsto 
  \begin{cases}
  f_2(x), & \text{if $x \in B$};\\
  h_j(x), & \text{if  $x \in \dom{(B_j)}$};\\
  f_1\circ h(x), & \text{otherwise}.  \end{cases}
\]

We first show that $f\circ H_{Q_2}= t$. 

\begin{align*}
f\circ H_{Q_2} &= f\circ (\restr{H_{Q_2}}{A_0}\cup \bigcup_{1\leq j\leq k} \restr{H_{Q_2}}{A_j})\\
      &= f\circ\restr{H_{Q_2}}{A_0} \cup \bigcup_{1\leq j \leq k}{f\circ\restr{H_{Q_2}}{A_j}}\\
      & =f_1\circ h\circ \restr{H_{Q_2}}{A_0} \cup \bigcup_{1\leq j \leq k}{h_j\circ\restr{H_{Q_2}}{A_j}}\\
      & = f_1\circ\restr{H_{Q_1}}{A_0}\cup \bigcup_{1\leq j \leq k}{\restr{t}{A_j}}= \restr{t}{A_0} \cup \bigcup_{1\leq j \leq k}{\restr{t}{A_j}} = t
\end{align*}

Finally, we show that $f(B_{Q_2})\subseteq I$.
 \begin{align*}
 f(B_{Q_2}) = f(B \cup \bigcup_{1\leq j\leq k} B_{j}) &= f(B)\cup \bigcup_{1\leq j\leq k} f(B_{j})\\ 
 &= f_2(B)\cup \bigcup_{1\leq j\leq k} h_j(B_{j}) \subseteq I
 \end{align*}
  \end{proof}
    To prove Theorem~\ref{theo:presv} we may thus limit ourselves to conjunctive queries with empty
    heads. First, we have a look at containments of the form $Q_1\subseteq Q_2$ where $B_{Q_1}$
    contains at least two non-redundant atoms. In what follows, when we write that a conjunctive
    query $Q$ is \emph{minimal}, we mean that $B_{Q}$ does not contain redundant atoms. (An atom in
    $B_Q$ is
    called \emph{redundant} if the query obtained from $Q$ by removing that atom is equivalent to $Q$.)
    \begin{lemma}\label{lem:BQ1two}
      Let $Q_1$ and $Q_2$ be $\CON$s where $Q_1$ is minimal and $H_{Q_1}=H_{Q_2}=()$. If $B_{Q_1}$ contains at least two atoms, then $Q_1\subseteq Q_2$ is equivalent to $\truef$ or is not monotone.
      \begin{proof}
         If $Q_1\subseteq Q_2$ is not equivalent to $\truef$, then $Q_1\not\sqsubseteq Q_2$. Thus,
         $Q_2(B_{Q_1})=\emptyset$ by Theorem~\ref{theo:cm}, whence we have $Q_1(B_{Q_1})\not\subseteq Q_2({B_{Q_1}})$. Since $|B_{Q_1}|\geq 2$, there exists a nonempty $B\subsetneq B_{Q_1}$. We have $Q_1(B)=\emptyset$ for otherwise $Q_1$ would not be minimal.

         Clearly, $Q_1(B)=\emptyset$ implies that $Q_1(B)\subseteq Q_2(B)$. Hence, $Q_1\subseteq Q_2$ is not monotone.
       \end{proof} 
    \end{lemma}

  We are now ready to prove Theorem~\ref{theo:presv}.
  \begin{proof}[Proof of Theorem~\ref{theo:presv}]
    Let $Q_1\subseteq Q_2$ be in $\lc{\SCON}\cap \MONO$. We want to show that $Q_1\subseteq Q_2$ is
    equivalent to $(()\leftarrow B)\neq \emptyset$ where $B$ is empty or $B$ consists of some of the
    connected components of $B_{Q_2}$.

    By Lemma~\ref{lem:contCQMONemptyheads} we may assume that $H_{Q_1}=H_{Q_2}=()$. We may
    furthermore assume that $Q_1$ is minimal. The constant $\truef$ query is expressed by $
    ()\leftarrow \emptyset\neq \emptyset$, so we may assume that $Q_1\not\sqsubseteq Q_2$. Thus,
    $Q_2(B_{Q_1})=\emptyset$ by Theorem~\ref{theo:cm}.

    If $B_{Q_1}$ contains at least two atoms, then $Q_1\subseteq Q_2$ is equivalent to $\truef$ by Lemma~\ref{lem:BQ1two}.

    If $B_{Q_1} = \emptyset$, then $Q_1\subseteq Q_2$ is equivalent to $Q_2\neq \emptyset$ which is in $\lns{\SCON}$.

    Finally, suppose that $B_{Q_1}$ contains exactly one atom. First, let us consider $B_{Q_1} = \{R(x_1,\ldots,x_n)\}$ where there is a repetition among $x_1,\ldots,x_n$. Define $I_1 = \{R(y_1,\ldots,y_n)\}$ where $y_1,\ldots,y_n$ are all different and not equal to any of $x_1,\ldots,x_n$. Clearly, $Q_1(I_1)=\emptyset$. 
    Since $Q_2(B_{Q_1})=\emptyset$, there is a connected component $C$ of $B_{Q_2}$ that does not map in $B_{Q_1}$. Furthermore, $C$ does not map into $I_1$ either, whence we also have $Q_2(I_1)= \emptyset$. Indeed, if $C$ would map into $I_1$, then $C$ would also map into $B_{Q_1}$ since $I_1$ maps into $B_{Q_1}$. It follows that $C$ does not map into $I_1\cup B_{Q_1}$ either, since $C$ is connected and $\adom{I_1}$ is disjoint from $\adom{B_{Q_1}}$. Therefore, $Q_2(I_1\cup B_{Q_1})=\emptyset$. Hence, $Q_1(I_1\cup B_{Q_1})\not\subseteq Q_2(I_1\cup B_{Q_1})$ since the head of $Q_1$ is in $Q_1(I_1\cup B_{Q_1})$. This contradicts that $Q_1\subseteq Q_2$ is monotone, since $Q_1(I_1)=\emptyset\subseteq Q_2(I_1)$. 

    So, the only body left to consider is $B_{Q_1}=\{R(x_1,\ldots,x_n)\}$ where $x_1,\ldots,x_n$ are
    all different and $R \in \dbsc$. Our proof now depends on the size of $\dbsc{}$.
  \begin{enumerate}
    \item Suppose that $\dbsc$ only contains the relation name $R$.  Then $Q_1(I)\neq \emptyset$ for
    any instance $I$ over $\dbsc{}$ since $B_{Q_1}=\{R(x_1,\ldots,x_n)\}$ where $x_1,\ldots,x_n$ are
    all different. Since $Q_1$ and $Q_2$ have empty heads, we may thus conclude that $Q_1\subseteq Q_2$ is equivalent to $Q_2\neq \emptyset$ in $\lns{\SCON}$.
    \item Suppose that $\dbsc$ only contains $R$ and exactly one other relation name $T$. Define $I_1 = \{T(y_1,\ldots,y_m)\}$ where $y_1,\ldots,y_m$ are different from each other and from $x_1,\ldots,x_n$. Since the body of $Q_1$ is an $R$-atom and $I_1$ only contains a $T$-atom, we have $Q_1(I_1)=\emptyset$. Hence, $Q_1(I_1)\subseteq Q_2(I_1)$. By the monotonicity of $Q_1\subseteq Q_2$, we also have $Q_1(I_1\cup B_{Q_1})\subseteq Q_2(I_1\cup B_{Q_1})$. Therefore, every connected component of $B_{Q_2}$ maps in $I_1$ or $B_{Q_1}$. Indeed, $Q_2(I_1\cup B_{Q_1})\neq \emptyset$ since the head of $Q_1$ is in $Q_1(I_1\cup B_{Q_1})$. This observation partitions the connected components of $B_{Q_2}$ into two sets $B'$ and $B''$, where $B'$ contains the components that map into $I_1$, and $B''$ contains the components that map into $B_{Q_1}$.   

    We now show that $Q_1\subseteq Q_2$ is equivalent to $Q' = () \leftarrow B'\neq \emptyset$. To this end, suppose that $Q'(I)\neq \emptyset$ and $Q_1(I)\neq \emptyset$ for some instance $I$ over $\dbsc{}$. Thus $B'$ and $B_{Q_1}$ map into $I$. Since $B''$ maps into $B_{Q_1}$ by construction, we also have that $B''$ maps into $I$. Hence, $Q_2(I)\neq \emptyset$ as desired. For the other direction, suppose that $Q_1(I)\subseteq Q_2(I)$ for some instance $I$ over $\dbsc{}$. If $Q_1(I)\neq \emptyset$, then $Q_2(I)\neq \emptyset$ by assumption. Clearly, $Q'(I)\neq \emptyset$ since $B_{Q'}$ is a subset of $B_{Q_2}$. On the other hand, if $Q_1(I)=\emptyset$, then $I$ has no $R$-facts. Since instances cannot be empty, it must contain at least one $T$-fact, so $I_1$ maps into $I$. Thus $B'$ also maps into $I$, whence $Q'(I)\neq \emptyset$ as desired.
    \item Finally, suppose that $\dbsc{}$ contains at least three relation names.
    Since $Q_2(B_{Q_1})=\emptyset$, there is a connected component $C$ of $B_{Q_2}$ that does not
    map into $B_{Q_1}$. In particular, we know that $C$ is not empty, whence it contains at least
    one atom, say a $T$-atom. (Note that $T$ might be equal $R$.) Since there are three relation names in $\dbsc{}$ there is at least one other relation name $S$ in $\dbsc{}$ that is not equal to $T$ or $R$. Define $I_2= \{S(z_1,\ldots,z_l)\}$ where  $z_1,\ldots,z_l$ are all different from each other and from $x_1,\ldots,x_n$. By construction, $C$ do not map into $I_2$ either, since $C$ contains an atom different from $S$. Thus, $Q_2(I_2\cup B_{Q_1})=\emptyset$, whence we have $Q_1(I_2\cup B_{Q_1})\not\subseteq Q_2(I_2\cup B_{Q_1})$ since $Q_1(I_2\cup B_{Q_1})\neq \emptyset$. However, $Q_1(I_2)=\emptyset$ since $R$ and $S$ are different, which implies that $Q_1(I_2)\subseteq Q_2(I_2)$. This contradicts the assumption that $Q_1\subseteq Q_2$ is monotone.
       \end{enumerate}
  \end{proof}
The proof of Theorem~\ref{theo:presv} gives us a procedure for deciding monotonicity for
containments of CQs.
\begin{corollary}
  Deciding whether a containment in $\lc{\SCON}$ is monotone is NP-complete.
  \begin{proof}
  Let $Q_1\subseteq Q_2$ be in $\lc{\SCON}$. By Lemma~\ref{lem:contCQMONemptyheads}
    we may remove the head variables of $Q_1$ and $Q_2$. The NP-hardness of our problem is taken
    care of by Lemma~\ref{lem:BQ1two}. Indeed, when $B_{Q_1}$ contains at least
    two non-redundant body atoms, the problem
    is equivalent to deciding $Q_1\sqsubseteq Q_2$, which is known to be NP-hard~\cite{cm}.

    Let us now show that the problem is in NP. By the proof of Theorem~\ref{theo:presv} we have the
    following cases when $Q_1$ is minimal:
    \begin{itemize}
      \item If $B_{Q_1}=\emptyset$, then $Q_1\subseteq Q_2$ is always monotone.
       \item If $|B_{Q_1}| \geq 2$, then $Q_1\subseteq Q_2$ is monotone if and only if $Q_1\sqsubseteq
       Q_2$ (Lemma~\ref{lem:BQ1two}).
       \item If $Q_1=\{R(x_1,\ldots,x_n)\}$ where there is a repetition among 
    $x_1,\ldots,x_n$, then $Q_1\subseteq
    Q_2$ is monotone if and only if $Q_1\sqsubseteq Q_2$.
      \item If $B_{Q_1} = \{R
    (x_1,\ldots,x_n)\}$ where $x_1,\ldots,x_n$ are all different, then:
    \begin{enumerate}
       \item[$(a)$] If $|\dbsc{}|=1$, then $Q_1\subseteq Q_2$ is always monotone;
       \item[$(b)$] If $|\dbsc{}|=2$, then $Q_1\subseteq Q_2$ is
    always monotone;
    \item[$(c)$] If $|\dbsc{}|\geq 3$, then $Q_1\subseteq Q_2$ is monotone if and only if $Q_1\sqsubseteq
    Q_2$.
     \end{enumerate} 
     \end{itemize}
     These properties suggest the following algorithm:
     \begin{enumerate}
      \item  Check if $B_{Q_1}=\emptyset$; if so, accept;
       \item Check if $Q_1\sqsubseteq Q_2$; if so, accept;
       \item Non-deterministically pick an atom $R(x_1,\ldots,x_n)$ in $B_{Q_1}$;
       \item Check the following:
       \begin{itemize}
         \item $()\leftarrow R(x_1,\ldots,x_n) \sqsubseteq Q_1$;
         \item $x_1,\ldots,x_n$ are all different.
       \end{itemize}
       \item Accept if $|\dbsc{}| \leq 2$ and the two checks above succeed; otherwise reject.
     \end{enumerate}
   The containment checks ($\sqsubseteq$) are well known to be in NP~\cite{cm}, so this algorithm is
   an NP algorithm.

If the algorithm accepts in step 1, then $Q_1\subseteq Q_2$ is equivalent to $Q_2\neq \emptyset$,
which is monotone. If the algorithm accepts in step 2, then the query $Q_1\subseteq Q_2$ is the constant
$\truef$ query, whence is trivially monotone. If the algorithm accepts in step 5, then the query $
()\leftarrow R(x_1,\ldots,x_n)$ is equivalent to $Q_1$, which is clearly minimal. Hence, by cases $
(a)$ and $(b)$ in the above properties, $Q_1\subseteq Q_2$ is monotone.

Conversely, suppose that $Q_1\subseteq Q_2$ is monotone. If $B_{Q_1}=\emptyset$ or $Q_1\sqsubseteq Q_2$, then
the algorithm accepts in step 1 or 2 respectively. Otherwise, consider a CQ $Q_1'$ obtained from $Q_1$ by omitting all
redundant atoms. Certainly, $Q_1'$ is minimal. Since $Q_1'\subseteq Q_2$ is monotone and
$Q_1'\not\sqsubseteq Q_2$, the above properties imply that $B_{Q_1'}$ consists of a single atom $R
(x_1,\ldots,x_n)$ where $x_1,\ldots,x_n$ are all different, and moreover that $|\dbsc{}| \leq 2$.
Hence, by picking this atom in step 3, the algorithm will accept. 
  \end{proof}
\end{corollary}
\section{Future Work}
There are several directions for future work. In this paper, conjunctive queries are not allowed to
have constants in the head and/or body. Our proof method does not work in the presence of
constants. Whether our characterization still holds in this case is still open. 

 Now that we have a syntactical characterization for monotone $\lc{\CON}$ we can look at other query languages. The first languages that come to mind are conjunctive queries with nonequalities, or negation, or unions. Another interesting language to consider is the more expressive first-order logic. When we allow infinite instances, the monotone first-order boolean queries are characterized by the positive first-order sentences with nonequalities~\cite{benedikt_book}. Whether this characterization still holds in restriction to finite instances remains open.

Another interesting line of work is to consider preservation theorems for other semantical properties, e.g., additivity. It can readily be verified that the additive queries in $\lns{\CON}$ are exactly those with connected bodies. Another example of a preservation theorem for additivity is: connected Datalog$^{\neg}$ captures the additive Datalog$^{\neg}$ queries under stratified semantics~\cite{frank_datalogcomponents_journal}.
\bibliographystyle{plain}

\begin{thebibliography}{1}

\bibitem{ahv_book}
S.~Abiteboul, R.~Hull, and V.~Vianu.
\newblock {\em Foundations of Databases}.
\newblock Addison-Wesley, 1995.

\bibitem{gurevich_monpos}
Miklos Ajtai and Yuri Gurevich.
\newblock Monotone versus positive.
\newblock {\em J. ACM}, 34(4):1004--1015, October 1987.

\bibitem{frank_datalogcomponents_journal}
Tom~J. Ameloot, Bas Ketsman, Frank Neven, and Daniel Zinn.
\newblock Datalog queries distributing over components.
\newblock {\em {ACM} Trans. Comput. Log.}, 18(1):5:1--5:35, 2017.

\bibitem{benedikt_book}
M.~Benedikt, J.~Leblay, B.~ten Cate, and E.~Tsamoura.
\newblock {\em Generating Plans from Proofs: The Interpolation-based Approach
  to Query Reformulation}.
\newblock Morgan\&Claypool, 2016.

\bibitem{cm}
A.K. Chandra and P.~Merlin.
\newblock Optimal implementation of conjunctive queries in relational data
  bases.
\newblock In {\em Proceedings 9th ACM Symposium on the Theory of Computing},
  pages 77--90. ACM, 1977.

\bibitem{changkeisler}
C.C. Chang and H.J. Keisler.
\newblock {\em Model Theory}.
\newblock North-Holland, 3rd edition, 1990.

\bibitem{gurevich_complexity}
Y.~Gurevich.
\newblock Toward logic tailored for computational complexity.
\newblock In M.M. Richter et~al., editors, {\em Computation and Proof Theory},
  volume 1104 of {\em Lecture Notes in Mathematics}, pages 175--216.
  Springer-Verlag, 1984.

\bibitem{rossman_hom}
Benjamin Rossman.
\newblock Homomorphism preservation theorems.
\newblock {\em J. ACM}, 55(3):15:1--15:53, August 2008.

\bibitem{stolboushkin_finmon}
Alexei~P. Stolboushkin.
\newblock Finitely monotone properties.
\newblock In {\em Proceedings of the 10th Annual IEEE Symposium on Logic in
  Computer Science}, LICS '95, pages 324--, Washington, DC, USA, 1995. IEEE
  Computer Society.

\end{thebibliography}

\end{document}